\documentclass[11pt]{article}

\usepackage{amsmath,amsfonts,amsthm}
\usepackage[numbers]{natbib}
\usepackage{color}
\usepackage{graphicx}
\usepackage{float}
\usepackage{bbold}
\usepackage{fancyhdr}
\usepackage{calc}
\usepackage{tikz}
\usetikzlibrary{decorations.markings}
\tikzstyle{vertex}=[circle, draw, inner sep=0pt, minimum size=6pt]
\newcommand{\vertex}{\node[vertex]}

\usepackage{bbm}

\newtheorem{thm}{Theorem}
\newtheorem{note}{Note}

\newtheorem{exa}{Example}
\newtheorem{defn}{Definition}

\newcommand{\mc}{\mathcal}

\newcommand{\hs}{\hspace{1mm}}

\fancypagestyle{plain}{ \footnotesize
\fancyfoot[L]{ \textit{Date}: \today \\
\textit{PNNL Information Release}: PNNL-26402}}

\title{Ubergraphs: A Definition of a Recursive Hypergraph Structure}
\begin{document}

\author{Cliff Joslyn, Kathleen Nowak}

\date{}
\fancypagestyle{empty}{
\fancyfoot[L]{Date: \today \\
PNNL Information Release: PNNL-26402}
}
\maketitle

\section{Introduction}

Partly in service of exploring the formal basis for Georgetown
University's AvesTerra database structure, we formalize a recursive hypergraph data structure, which we call an ubergraph. This type of data structure has been alluded to in passing but no formal explication exists to our knowledge. \footnote{https://en.wikipedia.org/wiki/Hypergraph\#Generalizations} As hypergraphs generalize graphs by allowing edges to have more than two vertices, ubergraphs generalize hypergraphs by allowing edges to contain other edges as vertices. Thus, all graphs are hypergraphs and all hypergraphs are ubergraphs.

The ability to do indirection in graph data structures by ``quoting''
or ``pointing to'' edges is absolutely central in graph-based data
science, and is accomplished in such systems by a variety of {\it ad
hoc} mechanisms such as reification. Hypergraphs are frequently used
as part of that armamentarium, but ubergraphs are a more robust representation framwork.

Note that here we deal only with undirected hyper- and
ubergraphs. Direction and/or orientation could prove very valuable,
but await further consideration \cite{GaGLoG93}.

\section{Hypergraphs}

A hypergraph is a generalization of a graph in which an edge can connect any number of vertices.

\begin{defn}
A \textbf{hypergraph} $H$ is a pair $(V,E)$ where $V$ is a set of vertices and $E \subseteq \mc{P}(V)$ is a set of non-empty subsets of $V$. The elements of $E$ are called hyperedges. 
\end{defn}


\begin{note}
An abstract simplicial complex is a hypergraph whose edge set is closed under subset.
\end{note}

\noindent The incidence matrix and Levi graph of a (hyper)graph express vertex-edge membership.

\begin{defn}
Let $H = (V,E)$ be a hypergraph with $|V| = n$ and $|E| = m$. The \textbf{incidence matrix} of $H$ is the $n \times m$ matrix $M$ defined by
$$
M_{ij} = 
\begin{cases}
1 & \text{if } v_i \in e_j \\
0 & \text{otherwise}.
\end{cases}
$$
\end{defn}

\begin{defn}
Let $H= (V,E)$ be a hypergraph with $|V| =n$ and $|E| =m$. 
Then the {\bf Levi graph} is the bipartite graph $G =
(V \hs \dot\cup \hs E, E')$, where $(v_i,
e_j) \in E'$ if and only if $v_i \in e_j$.
\end{defn}

\begin{exa}

Let $H$ be the hypergraph with vertex set $V = \{1,2,3,4,5\}$ and edge
set 
	\[ E = \{\{1\},\{1,3\}, \{2,3\}, \{1,3,5\}\}.	\]
	The incidence matrix for $H$ is 
$$
M = 
\begin{bmatrix}
1 & 1 & 0 & 1 \\
0 & 0 & 1 & 0 \\
0 & 1 & 1 & 1 \\
0 & 0 & 0 & 0 \\
0 & 0 & 0 & 1
\end{bmatrix}
$$
and the Levi graph representation is 

\[\begin{tikzpicture}[
	every edge/.style={
        draw,
        postaction={decorate,
                    decoration={markings,mark=at position 1 with {\arrow{>}}}
                   }
        }
        ]
	\vertex[fill] (s5) at (0,1) [label=left:$v_{5}$] {};
	\vertex[fill] (s4) at (0,2) [label=left:$v_{4}$] {};
	\vertex[fill] (s3) at (0,3) [label=left:$v_{3}$] {};
	\vertex[fill] (s2) at (0,4) [label=left:$v_{2}$] {};
	\vertex[fill] (s1) at (0,5) [label=left:$v_{1}$] {};
	\vertex[fill] (t4) at (1,1.5) [label=right:$e_{4}$] {};
	\vertex[fill] (t3) at (1,2.5) [label=right:$e_{3}$] {};
	\vertex[fill] (t2) at (1,3.5) [label=right:$e_{2}$] {};
	\vertex[fill] (t1) at (1,4.5) [label=right:$e_{1}$] {};

	\path
	 	
	 	(s1) edge (t1)
	 	(s1) edge (t2)
	 	(s1) edge (t4)
	 	(s2) edge (t3)
	 	(s3) edge (t2)
	 	(s3) edge (t3)
	 	(s3) edge (t4)
	 	(s5) edge (t4)
	 	
	;
\end{tikzpicture}\]

\end{exa}

Where the Levi graph represents the vertex-edge membership relation
$\in$, since $E \subseteq {\cal P}(V)$ is a set system of $V$, it also
manifests edge-edge inclusion as a partial order on the edges through
$\subseteq$. This can also be read from the Levi graph as $N^-(e_i) \subset N^-(e_j)$ if and only if $e_i \subset e_j$. In our example, we have $e_1 \subset e_2 \subset e_4$,
while $e_3$ is non-comparable.

\section{Ubergraphs}

One way to generalize hypergraphs is to allow edges to contain not
only vertices but other edges. For a finite set $X$, we define 
$$
\mc{P}(X)^k = \mc{P}\left(\bigcup_{i = 0}^k P_i\right), \hs \text{ where } P_0 = X \hs \text{ and } P_i = \mc{P} \left(\bigcup_{j = 0}^{i-1} P_j \right) \text{ for } i \geq 1.
$$
\begin{defn}

A \textbf{depth $k$ ubergraph $U$} is a pair $(V,E)$ where $V$ is a set of fundamental
vertices and $E \subseteq \mc{P}(V)^k$ is a finite set of uberedges. Additionally, if $s \notin V$ belongs to an edge, we require that $s$ is itself an edge.
\end{defn}

\begin{note}
Every hypergraph is a depth $0$ ubergraph.  
\end{note}

Since uberedges are allowed to contain other edges, we call the elements of 
$$
V \hs \cup \hs \left(\bigcup_{e \in E} e \right) 
$$
\textit{vertices} and the elements of $V$ \textit{fundamental vertices}.

Let $U = (V,E)$ be an ubergraph with $|V| = n$ and $|E| = m$. The \textit{incidence matrix} of this type of hypergraph is a matrix $M$ of order $(n+m) \times m$ where 
$$
M_{xy} = 
\begin{cases}
1 & \text{if } x \in y \\
0 & \text{otherwise}.
\end{cases}
$$
Moreover, the Levi graph of a hypergraph
generalizes to what we will call the \textit{uber-Levi graph} of
$U$ to express uberedge membership. It has one vertex corresponding to each fundamental vertex and edge of $U$
and a directed edge from $x$ to $y$ if $x$ is a member of $y$ in $U$.

\begin{exa}
Let $U$ be the ubergraph with fundamental vertex set $V = \{1,2,3\}$ and edges 
\begin{eqnarray*}
E &=& \{e_1, e_2, e_3, e_4, e_5\} \\
&=& \{\{1\}, \{1,3\}, \{1,3,e_1\}, \{2,e_2\}, \{1,e_4\}\} \\
&=& \{\{1\}, \{1,3\}, \{1,3,\{1\}\}, \{2, \{1,3\}\}, \{1, \{2, \{1,3\}\}\} 
\end{eqnarray*}

The incidence matrix for $H$ is 
$$
M = 
\begin{bmatrix}
1 & 1 & 1 & 0 & 1\\
0 & 0 & 0 & 1 & 0 \\
0 & 1 & 1 & 0 & 0\\
0 & 0 & 1 & 0 & 0\\
0 & 0 & 0 & 1 & 0\\
0 & 0 & 0 & 0 & 0\\
0 & 0 & 0 & 0 & 1\\
0 & 0 & 0 & 0 & 0 
\end{bmatrix}
$$
and the uber-Levi graph representation is 
 
\begin{figure}[H]
\centering
\includegraphics[scale=.9]{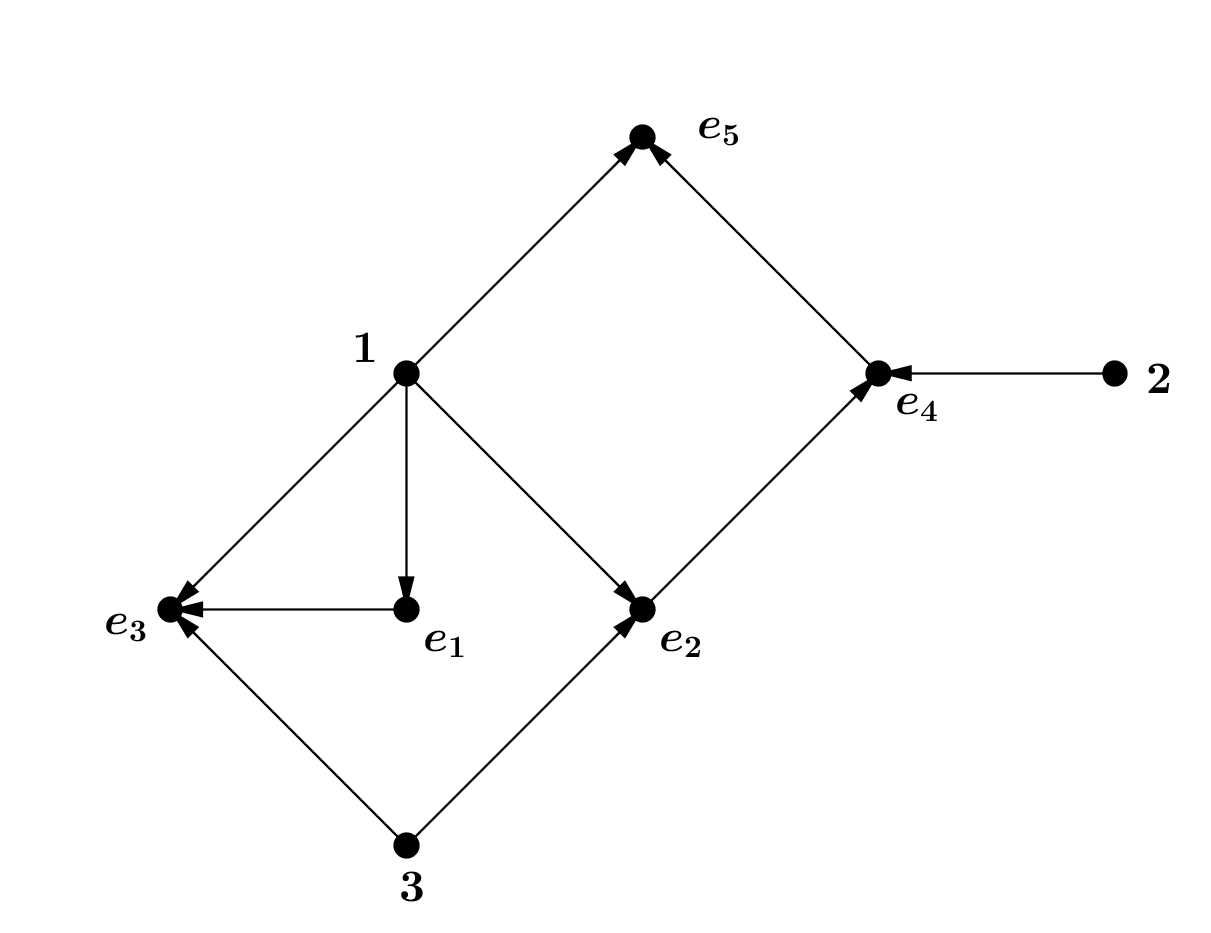}
\end{figure}

\end{exa}

\begin{note}
The uber-Levi graph is a directed acyclic graph (DAG). The
roots (vertices with no in-neighbors) correspond to the fundamental vertices of $U$,
and the vertices with positive in degree correspond to the edges of
$U$. Moreover, every $DAG$ yields an ubergraph.  
\end{note}

Where the uber-Levi graph represents the vertex-edge membership relation $\in$, since $E \subseteq  \mc{P}\big( \bigcup_{i=0}^\infty P_i \big)$ is a set system of $\bigcup_{i=0}^\infty P_i$, it also manifests edge-edge inclusion as a partial order on the edges through $\subseteq$. As before, this can be read from the uber-Levi graph since $e_i \subset e_j$ if and only if $N^-(e_i) \subset N^-(e_j)$. In our example, we have $e_1 \subset e_2 \subset e_3$ and $e_1 \subset e_5$, while $e_4$ is non-comparable. 

As constructed so far, ubergraphs can express the syntax of
generalized graph data structures like AvesTerra, where nodes can have
a variable number of attributes, and in turn those attributes can be
other nodes. In our example, we have $e_1 \in e_3$: the uberedge
$e_3$ has another uberedge $e_1$ as an element.

Additionally, it has been discussed that AvesTerra
may wish to model situations where edges can refer to 
themselves, either directly or indirectly. This corresponds to dropping the requirement: If $s \in P_i, i > 0$, belongs to an edge, then $s$ is itself an edge. In our example, if we were
to change the definition of $e_5$ from $e_5 = \{1,e_4\}$ to $e'_5 = \{1,e_4, e_2 \}$, then the uber-Levi graph would change to 
\begin{figure}[H]
\centering
\includegraphics[scale=.9]{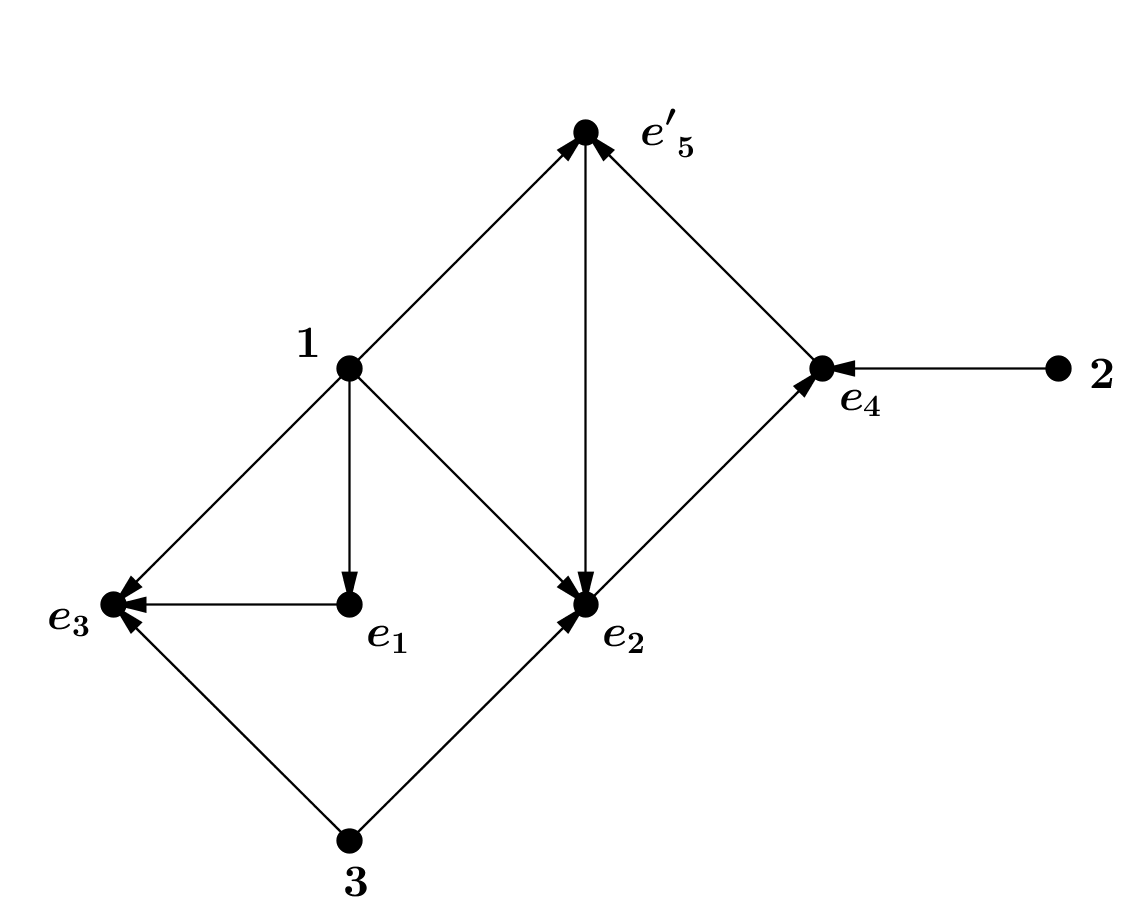}
\end{figure}

Note the inclusion of a cycle, making the uber-Levi graph now a
general directed graph. Allowing (ultimately)
expressions like $e = \{ e \}$, and thus arbitrary cycles in the
uber-Levi graph,  violates the axiom of foundation. The vertex set is
no longer well defined, and non-well-founded sets would need to be
invoked.

\subsection{Basic Concepts}

Ubergraphs are a generalization of hypergraphs, hence many of the definitions of hypergraphs carry verbatim to ubergraphs. Most of the vocabulary given here is generalized from \cite{Bretto}.

Let $U = (V,E)$ be a depth $k$ ubergraph with $|V| = n$ and $|E| = m$. Let $M$ be the incidence matrix of $U$. By definition the \textit{empty ubergraph} has $V = E = \emptyset$ and we call any ubergraph with $V \neq \emptyset$ and $E = \emptyset$ a \textit{trivial ubergraph}.  For $e \in E$, we define 
$$
V^0(e) = \bigcap_{\substack{S \subseteq V, \\ e \subseteq \mc{P}(S)^k}} S 
$$
to be the minimum set of fundemental vertices in an ubergraph containing $e$.
Then we have the following analogs of subgraph:
\begin{itemize}
\item For $E' \subseteq E$, the ubergraph $(V, E')$ is called a \textit{sububergraph}.
\item For $V' \subseteq V$, the \textit{induced sububergraph} $U[V']$ of the ubergraph $U$ is the ubergraph $U(V') = (V', E')$ where
$$
E' = \{e \in E \hs | \hs V^0(e) \subseteq V'\}.
$$
\end{itemize}
\noindent Let $U$ be the ubergraph with vertices $V = \{1,2,3,4,5\}$ and edge set 
$$
E = \{\{1,2\}, \{1,\{1,2\}\}, \{\{3\},\{\{1,4\}\}\}, \{1,4,5\}\}.
$$
Then $U' = (\{1,2,3,4,5\}, \{\{1,2\},\{1,4,5\}\})$ is a sububergraph of $U$ and $U[\{1,2\}] = (\{1,2\}, \{\{1,2\},\{1,\{1,2\}\}\})$.

Two vertices $x,y$ of an ubergraph are \textit{adjacent} if there is an uberedge which contains both elements. Two uberedges are \textit{incident} if their intersection is not empty. The \textit{degree} of a vertex $x$ is the number of uberedges containing $x$. 

Let $x, y \in V \hs \cup \hs E$. A \textit{path} $P$ from $x$ to $y$ is a sequence 
$$
x = x_1, e_1, x_2, e_2, \dots, x_s, e_s, x_{s+1} =y
$$
such that 
\begin{itemize}
\item $x_1, x_2, \dots, x_{s+1}$ are all distinct vertices except possibly $x_1 = x_{s+1}$,
\item $e_1, e_2, \dots, e_s$ are distinct uberedges, and
\item $x_i, x_{i+1} \in e_i$ for all $i = 1, 2, \dots, s$.
\end{itemize}

If $x=y$ the path is called a \textit{cycle}. The integer $s$ is the \textit{length} of the path. We say that $U$ is \textit{connected} if for every pair of vertices, there is a path which connects these vertices; otherwise we describe $U$ as \textit{disconnected}. 

\subsection{Matrices, Ubergraphs, and Entropy}

Let $U = (V,E)$ be an ubergraph with $|V| =n$ and $|E| = m$. As stated earlier, the incidence matrix of an ubergraph is an $(n+m) \times m$ matrix $M$ where
$$
M_{xy} = 
\begin{cases}
1 & \text{if } x \in y \\
0 & \text{otherwise}.
\end{cases}
$$

Many basic concepts can be computed from the incidence matrix. For example, the degree of a vertex $x$ is 
$(M \mathbbm{1})_x$ and two uberedges are incident if and only if the inner product of the corresponding rows of $M$ is nonzero. Further, the incidence matrix of any (induced) sububergraph is a submatrix of $M$.

Next, we define the \textit{adjacency matrix} $A(U)$ of $U$ to be the square matrix whose rows and columns are indexed by the fundamental vertices and edges of $U$ such that for all
$x, y \in V \hs \cup \hs E$,
$$
A_{xy} = 
\begin{cases}
|\{e \in E \hs | \hs x,y \in e\}| & \text{if } x \neq y \\
0 & \text{otherwise}.
\end{cases}
$$
Let $D(x) = \sum\limits_{y \in V \hs \cup \hs E} a_{x,y}$. Then the \textit{Laplacian matrix} of $U$ is given as 
$$
L(U) = D -A(U) 
$$
where $D = diag(D(v_1), \dots, D(v_n), D(e_1), \dots, D(e_m))$. Note that $L(U)$ is Hermitian so it's eigenvalues are real. Further, by an application of the Gershgorin disk Theorem, they must be nonnegative. Since $\sum\limits_{i=1}^{n+m} \lambda_i = Tr(L(U)) =  \sum\limits_{i=1}^{n+m} D(x_i) := \hat{d}$, we have that 
$$
\left(\mu_i\right)_{i = 1}^{n+m} := \left(\frac{\lambda_i}{\hat{d}}\right)_{i = 1}^{n+m}
$$
is a discrete probability distribution. Thus, we can define the \textit{algebraic ubergraph entropy} of $U$ by
$$
I(U) = -\sum_{i = 1}^{n+1} \mu_i \log_2(\mu_i).
$$

\subsection{Similarity and Metric on Ubergraphs}

When we have two structures, one of the most important tasks is to compare them. This comparison is done with an isomorphism. 

\begin{defn}
Let $U = (V,E)$ and $U' = (V',E')$ be two ubergraphs. We say that $U$ and $U'$ are \textbf{isomorphic}, denoted $U \simeq U'$, if there exists a bijection
$$
\varphi: V \rightarrow V'
$$
such that 
$$
e \in E \text{ if and only if } \varphi(e) := \{\varphi(x) \hs | \hs x \in e\} \in E'.
$$
\end{defn}

This similarity measure is manifested in both the uber-Levi graph and the incidence matrix. 

\begin{thm}
Two ubergraphs are isomorphic if and only if their uber-Levi graphs are isomorhpic.
\end{thm}

\begin{proof}
Let $U$ and $U'$ be ubergraphs with uber-Levi graphs $D$ and $D'$ respectively. First suppose that $U \simeq U'$ and let $\varphi: V(U) \rightarrow V(U')$ be an isomorphism. Define 
$$
\psi: V(D) \rightarrow V(D') \text{ by } \psi(v) = \varphi(v) \text{ and } \psi(e) = \varphi(e) = \{\varphi(x) \hs | \hs x \in e\}. 
$$
Then 
$$
(x,y) \in E(D) \stackrel{def}{\Leftrightarrow} x \in y  \stackrel{\simeq}{\Leftrightarrow} \psi(x) \in \psi(y) \stackrel{def}{\Leftrightarrow} (\psi(x),\psi(y)) \in E(D').
$$

Now suppose that $D \simeq D'$ and let $\psi: V(D) \rightarrow V(D')$ be an isomorphism. We claim that $\psi \vert_{V(U)}$ is an isomorphism from $U$ to $U'$. First note that $v$ is a fundamental vertex if and only if it's in-degree in the uber-Levi graph representation is zero. Thus, since isomorphisms preserve degree, we must have that $\psi\vert_{V(U)}$ is bijection from $V(U)$ to $V(U')$. Further,
$$
x \in e \stackrel{def}{\Leftrightarrow} (x,e) \in E(D) \stackrel{\simeq}{\Leftrightarrow} (\psi(x),\psi(e)) \in E(D') \stackrel{def}{\Leftrightarrow} \psi(x) \in \psi(e).
$$
Thus, by recursively applying this argument we have that 
$\psi\vert_{V(U)}$ is a bijection from $V(U)$ to $V(U')$ such that 
$e \in E(U)$ if and only if $\psi\vert_{V(U)}(e) \in E(U')$. 
\end{proof}

%

\bibliographystyle{plain}

\bibliography{uberbib1.bib}

\end{document}